\documentclass[conference]{IEEEtran}
\def\IEEEtitletopspace{18pt}
\IEEEoverridecommandlockouts
\usepackage{cite}
\usepackage{amsmath,amssymb,amsfonts,amsthm}
\usepackage{algorithmic}
\usepackage{graphicx,subcaption}
\usepackage{textcomp}
\usepackage{comment,lipsum}
\usepackage[margin=.75in]{geometry}

\usepackage[utf8]{inputenc}

\usepackage{tikz}
\usetikzlibrary{arrows.meta,decorations.pathmorphing,backgrounds,positioning,fit,petri}

\def\x{{\mathbf x}}
\def\u{{\mathbf u}}
\def\v{{\mathbf v}}

\theoremstyle{theorem}
\newtheorem{theorem}{Theorem}
\newtheorem{corollary}{Corollary}
\newtheorem{lemma}{Lemma}
\theoremstyle{definition}
\newtheorem{definition}{Definition}

\newtheorem*{problem*}{Problem Statement}

\theoremstyle{remark}
\newtheorem{remark}{Remark}
\newtheorem{example}{Example}

\usepackage{pgfplots} 
\pgfplotsset{compat=1.18}
\usepackage{xcolor}
\def\BibTeX{{\rm B\kern-.05em{\sc i\kern-.025em b}\kern-.08em
    T\kern-.1667em\lower.7ex\hbox{E}\kern-.125emX}}
    
\begin{document}
\title{Control Design for Trajectory Tracking and Stabilization of Sensor LOS in an Inertially Stabilized Platform
}

\author{\IEEEauthorblockN{Abinash Agasti}
\IEEEauthorblockA{\textit{Department of Electrical Engineering} \\
\textit{Indian Institute of Technology Madras}\\
Chennai, India \\
abinashagasti@outlook.com}
\and
\IEEEauthorblockN{Angana Hazarika}
\IEEEauthorblockA{\textit{Department of Electrical Engineering} \\
\textit{National Institute of Technology Kurukshetra}\\
Kurukshetra, India \\
anganahazarika2712@gmail.com }
\and
\IEEEauthorblockN{Dr. Bharath Bhikkaji}
\IEEEauthorblockA{\textit{Department of Electrical Engineering} \\
\textit{Indian Institute of Technology Madras}\\
Chennai, India \\
bharath@ee.iitm.ac.in}
}

\maketitle

\begin{abstract}
Optical sensors are often mounted on moving platforms to aid in a variety of tasks like data collection, surveillance and navigation. This necessitates the precise control of the inertial orientation of the optical sensor line-of-sight (LOS) towards a desired stationary or mobile target. A two-axes gimbal assembly is considered to achieve this control objective which can be decomposed into two parts - stabilization and tracking. A novel state space model is proposed based on the dynamics of a two-axes gimbal system. Using a suitable change of variables, this state space model is transformed into an LTI system. Feedback linearization based control laws are proposed that achieve the desired objectives of stabilization and tracking. The effectiveness of these control laws are demonstrated via simulation in MATLAB based on a typical model of a two-axes gimbal system. 
\end{abstract}
\section{Introduction}  

Optical sensors such as IR, radar and camera are often mounted on moving platforms like ground vehicles, aircraft or marine vessels to collect data, conduct surveillance or aid in navigation. Such uses can be widely observed in applications spanning military reconnaissance, agricultural monitoring, or guided weapons systems. In these applications it is imperative to precisely control the inertial orientation of the optical sensor line-of-sight (LOS) towards a desired stationary or mobile target. An inertially stabilized platform (ISP), typically achieved through gimbal assemblies, is an appropriate mechanism to achieve this desired control \cite{kennedy_2003,debruin_2008,precise_pointing_aerospace_2023}. 

An ISP comprises of three modules: an electromechanical assembly that interfaces between the optical sensor and the platform body, a control system that orients the sensor in the desired direction, and auxiliary equipment that determines the desired target location \cite{masten_2008}. The focus of this paper lies in the design of the control system for the ISP, which typically consists of two subsystems - the inner stabilization loop and the outer tracking loop \cite{hilkert_2008}. The objective of the stabilization loop is to maintain the inertial orientation of the optical sensor in order to obtain jitter-free high quality data. The inner loop obtains this objective by controlling the angular rates of the sensor LOS. Meanwhile, the objective of the outer tracking loop is to orient the sensor LOS towards the desired target. Based on the desired target location, the outer loop provides desired rate commands to the inner loop. The cascading of these subsystems allows the optical sensor LOS to be oriented towards the desired target and collect high quality jitter-free data. 

To achieve LOS stabilization and tracking, the gimbal system must counteract all torque disturbances while orienting the sensor LOS towards the desired target. The torque disturbances originate from three primary sources: platform body motion, cross-coupling disturbances, and gimbal system mass unbalance.  This necessitates a precise mathematical model for the system. Equations of motion for each gimbal axis can be derived using both Newton's second law as well as Lagrange equations. In the work by Ekstrand \cite{ekstrand_2002}, the equations were derived with the assumption of the gimbal system having no mass unbalance. Moreover, this work also explored a special case of symmetry that resulted in vastly reduced cross-coupling disturbances. Other studies have considered the effects of dynamic mass unbalance and asymmetry in the modelling of the two-axes gimbal configuration \cite{cross_coupling_2013,double_gimbal_2020,strong_coupling_2019}. 

Until recently, LOS stabilization and tracking have been accomplished using classical control methods, often variations of the well-established PID controller design
\cite{classical_pid_2021, classical_pid_2017}. Nevertheless, there have been efforts to explore advanced  modern control techniques to enhance stabilization. In \cite{lqg_stab_2006}, an LQG/LTR controller was employed for the inner stabilization loop design. In other instances, adaptive versions of the PID controller have been utilized to achieve the control objective \cite{selftuning_pid_2014, anfis_2023}. Sliding mode control has also emerged as another alternative for achieving sensor LOS stabilization \cite{sliding_mode_2020, sliding_mode_2023}. 
However, there has been limited research dealing with the control objectives of stabilization and tracking in an ISP using the methods of nonlinear analysis and control.

This paper introduces a nonlinear control approach for LOS stabilization and tracking. Unlike the conventional approach of employing separate controllers for the inner and outer loops, this work introduces a unified control law. First, a control law is introduced to achieve stabilization. This is followed by another control law that enables seamless LOS tracking without the need of an inner loop. The main contributions of this paper can be surmised in the following points:
\begin{enumerate}
    \item First, a novel state space model for the two-axes gimbal system is developed by making an appropriate choice for the state variables, assuming the gimbal to be symmetric and having no mass unbalance.
    \item The choice of state variables renders the dynamical system as a set of two decoupled Brunovsky canonical subsystems (a series of integrators, see Chapter 13 \cite{Khalil_2002}). This further allows a transformation into a linear dynamical system using a suitable change of variables.
    \item Leveraging this transformed linear system, a feedback control law is designed to facilitate exponential convergence to a desired angular velocity trajectory. Stabilization is obtained as a consequence of this result.
    \item Additionally, a feedback control law is designed that achieves exponential convergence to a desired LOS trajectory, thus achieving LOS tracking.
    \item The efficacy of these control laws is then demonstrated on a gimbal system model using MATLAB simulation.
\end{enumerate}
 
This paper is organized as follows: The comprehensive model of a two-axes gimbal system along with the control objectives are provided in Section \ref{sec:dynamics}. In Section \ref{sec:model}, a novel state space model is proposed, and a change of variables is considered which transforms the nonlinear state space model into a linear system. The nonlinear state feedback control laws for the stabilization and tracking of sensor LOS are proposed in Section \ref{sec:control_law}. The simulation of the proposed control laws has been implemented on MATLAB based on the model of a typical two-axes gimbal system in Section \ref{sec:sim}. Finally in Section \ref{sec:conc}, the conclusions are drawn and possible future directions are cited.

\section{Two-axes Gimbal Dynamics}
\label{sec:dynamics}
Consider a two-axes yaw-pitch gimbal system as shown in Figure \ref{fig:gimbal}. The outer gimbal is the yaw gimbal while the inner gimbal is the pitch gimbal. The optical sensor is placed on the inner gimbal. A rate gyro is placed on the platform to measure its angular velocity with respect to an inertial frame of reference. 
\begin{figure}[htp]
    \centering
    \includegraphics[scale=0.7]{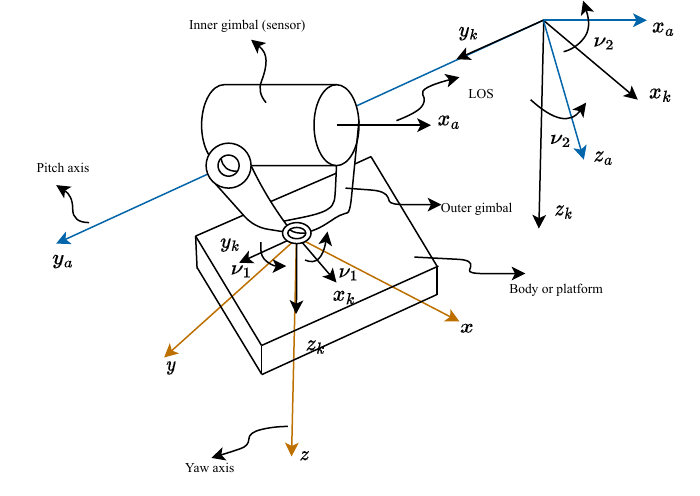}
    \caption{Two-axes gimbal system}
    \label{fig:gimbal}
\end{figure} 

Three coordinate frames are introduced as follows: frame $B$ fixed to the platform body with axes $(x,y,z)$, frame $K$ fixed to the yaw (outer) gimbal with axes $(x_k,y_k,z_k)$ and frame $A$ fixed to the pitch (inner) gimbal with axes $(x_a,y_a,z_a)$. The $x_a$ axis coincides with the sensor optical axis. The center of rotation is assumed to be at the common origin of all the coordinate frames, \textit{i.e.} the gimbals are considered to be rigid bodies with no mass unbalance.\par

The inertial angular velocity vectors of frames $B, K,$ and $A$, respectively are as follows:
\begin{equation}
    \Vec{\omega}_B=\begin{bmatrix}
         p\\
         q\\
         r
  \end{bmatrix} ;
    \Vec{\omega}_K=\begin{bmatrix}
       p_k\\
       q_k\\
       r_k
  \end{bmatrix} ;
    \vec{\omega}_A=\begin{bmatrix}
       p_a\\ 
       q_a\\ 
       r_a
  \end{bmatrix} \label{eq:ang_vel}
\end{equation}
where $p, q, r$ are the components of body angular velocities of frame $B$ in relation to inertial space about $x, y, z$ axes respectively. Similarly, $p_k, q_k, r_k$ are the yaw gimbal angular velocity components in relation to inertial space about $x_k, y_k, z_k$ axes respectively and $p_a, q_a, r_a$ are the pitch gimbal angular velocity components in relation to inertial space about $x_a, y_a, z_a$ respectively. Further, let $\theta_q(t):=\int_0^tq_a(\tau)d\tau$ and $\theta_r(t):=\int_0^tr_a(\tau)d\tau$ represent the elevation and azimuth angles of the sensor LOS respectively with respect to an inertial frame of reference. Here onwards, all the angular velocity terms and $\theta_q,\theta_r$ are used with an implicit dependence on time.\par

Now, let $\nu_1$ denote the angle of rotation about the $z$-axis to carry body-fixed frame $B$ into coincidence with the yaw gimbal frame $K$. Similarly, let $\nu_2$ denote the angle of rotation about the $y_k$-axis to carry the yaw gimbal frame $K$ into coincidence with the pitch gimbal frame $A$. As a result, the coordinate transformation matrices from frame $B$ to $K$ and $K$ to $A$ in terms of $\nu_1$ and $\nu_2$ are defined as follows:  
\begin{equation}
   R_{KB}=\begin{bmatrix}
       \cos\nu_1 & \sin\nu_1 & 0\\
       -\sin\nu_1 & \cos\nu_1 & 0\\
       0 & 0 & 1
   \end{bmatrix} 
\end{equation}
\begin{equation}
   R_{AK}=\begin{bmatrix}
       \cos\nu_2 & 0 & -\sin\nu_2\\
        0 & 1 & 0\\
        \sin\nu_2 & 0 & \cos\nu_2
    \end{bmatrix}
\end{equation} 

Utilizing the angular velocity vectors given in \eqref{eq:ang_vel}, along with the coordinate transformation matrices defined above, the following angular velocity relations can be obtained:
\begin{equation}
    \Vec\omega_K-R_{KB}\Vec\omega_B=\begin{bmatrix}
        0\\0\\\dot\nu_1
    \end{bmatrix}, \text{ and }
    \Vec\omega_A-R_{AK}\Vec\omega_K=\begin{bmatrix}
        0\\\dot\nu_2\\0
    \end{bmatrix}.
\end{equation}
Simplifying the above equation, the inertial angular velocity of frame $K$ can be written in terms of the inertial angular velocity of frame $B$ as 
\begin{equation}
  \begin{aligned}   
     & p_k= p\cos\nu_1 + q\sin\nu_1,    \\ 
     & q_k= -p\sin\nu_1 + q\cos\nu_1,   \\    
     & r_k= r + \dot{\nu_1},
  \end{aligned} \label{eq:b2k}
\end{equation}  
and the inertial angular velocity of frame $A$ can be written in terms of the inertial angular velocity of frame $K$ as
\begin{equation}
   \begin{split}
          & p_a= p_k\cos\nu_2 - r_k\sin\nu_2,  \\
          & q_a=q_k + \dot{\nu_2},   \\
          & r_a=  p_k\sin\nu_2 + r_k\cos\nu_2.
    \end{split} \label{eq:k2a}
\end{equation}  
Now, consider the inertia matrices of the pitch and yaw gimbal as \begin{equation} 
\text{Inner gimbal}:           J_A=\begin{bmatrix}
                        J_{ax} & D_{xy} & D_{xz}\\
                         D_{xy} & J_{ay} & D_{yz}\\
                         D_{xz} & D_{yz} & J_{az}
                \end{bmatrix}
\end{equation}
\begin{equation}
\text{Outer gimbal}:          J_K=\begin{bmatrix}
                    J_{kx} & d_{xy} & d_{xz}\\
                    d_{xy} & J_{ky} & d_{yz}\\
                    d_{xz} & d_{yz} & J_{kz}
                \end{bmatrix}
\end{equation}
where the diagonal terms represent the moments of inertia while the off diagonal terms represent the products of inertia.\par

The equations of motion can then be obtained either via the Newton's equations of motion or through the Lagrange equations. As shown in \cite{ekstrand_2002}, this results in two equations of motion for the two gimbal axes. Assuming the total external torque about the pitch gimbal $y_a$-axis to be $T_{y}$, the equation of motion for the pitch gimbal can be derived as
\begin{equation}
    J_{ay}\dot{q _a} = T_y + T_{D_y},\label{eq:pitch_eqn}
\end{equation}
where $T_{D_y}$ represents all the inertia disturbances that arise along the $y_a$-axis of the gimbal system. These inertia disturbance terms are functions of the inertia matrices, the angular velocity vector of frame $B$ and the angles of rotation $\nu_1$ and $\nu_2$.

A similar equation can be written for the yaw gimbal motion given by
\begin{equation}
    J_k\dot r_k=T_z+T_{D_z}, \label{eq:yaw_eqn}
\end{equation}
where $J_k$ is an inertia term dependent on the angle of rotation $\nu_2$, $T_z$ is the total external torque about the yaw gimbal $z_k$-axis, and $T_{D_z}$ contains all the inertia disturbances due to the design of the gimbal system.

Due to the inherent characteristics of the gimbal system and the coupling between the yaw and pitch gimbals, the system is susceptible to various inertia disturbances. However, system design typically aims to minimize these disturbances. One effective approach is to design a symmetric gimbal system without any mass unbalance. The symmetric assumptions are entailed by the following conditions:
\begin{equation} 
   \begin{split}
    & D_{xy} = D_{xz} = D_{yz} = 0, \\
    & J_{ax} = J_{az}, \\
    & d_{xy} = d_{xz} = d_{yz} = 0, \\ 
    & J_{kx}+J_{ax}=J_{ky}.  \\
   \end{split} \label{eq:assum}
\end{equation} 
Under the assumptions \eqref{eq:assum} we have $J_k=J_{kz}+J_{az}$, and the equations of motion for the pitch and yaw gimbals can be written as 
\begin{equation}
    \begin{aligned}
        J_{ay}\dot q_a&=T_y, \text{ and}\\
        J_k\dot r_k&=T_z-J_{ay}p_kq_a.
    \end{aligned}\label{eq:dynamics}
\end{equation}
While the assumptions in the gimbal system design \eqref{eq:assum} remove most of the identified noise sources in the system, it is possible for design errors to persist resulting in some noise entering the system through the dynamic equations in \eqref{eq:dynamics}.

Note that the external torque is provided primarily by the motors placed in each gimbal axis. With an understanding of the gimbal system's dynamics, the control objectives of stabilization and tracking can be mapped onto the system variables in the following manner:
\begin{enumerate}
    \item \textbf{Stabilization - } 
    The primary objective of LOS stabilization hinges on the isolation of torque disturbances affecting the pitch and yaw axes. This is achieved by driving the angular velocities $q_a$ and $r_a$ to zero. 
    \item \textbf{Tracking - }
    The objective of tracking refers to the sensor LOS oriented towards the desired target, while rejecting the torque disturbances. This is obtained by the sensor LOS elevation ($\theta_q$) and azimuth ($\theta_r$) angles following a desired target trajectory.
\end{enumerate}

\section {State Space Modelling} \label{sec:model}
Given the dynamic equations of a two-axes gimbal system \eqref{eq:dynamics}, consider a state space model described by a state vector $\x\in\mathcal X:= S^1\times\mathbb R\times S^1\times\mathbb R$ and a control vector $\u\in\mathbb R^2$ given by 
\begin{equation}
   \x=\begin{bmatrix}
       x_1\\
       x_2\\
       x_3\\
       x_4
   \end{bmatrix}  =\begin{bmatrix}
             \nu_2 \\
             \dot{\nu_2}\\
             \nu_1\\
             \dot{\nu_1}
    \end{bmatrix},\quad 
    \u=\begin{bmatrix}
        u_1\\ u_2
    \end{bmatrix}=\begin{bmatrix}
        T_{y}\\ T_{z}
    \end{bmatrix}.
\end{equation}
Clearly, $\dot x_1=x_2$ and $\dot x_3=x_4$. From \eqref{eq:k2a}, $x_2=\dot\nu_2=q_a-q_k$, and using \eqref{eq:dynamics}, we have
\begin{equation}
    \dot x_2=\frac{u_1}{J_{ay}}+f_1(t,\x), \label{eq:x2_dyn}
\end{equation}
where $f_1(t,\x)=\dot{p}\sin{x_3}+ x_4p\cos{x_3}- \dot{q}\cos{x_3}+ x_4q\sin{x_3}$ is obtained by differentiating \eqref{eq:b2k}. Similarly, from \eqref{eq:b2k}, $x_4=\dot\nu_1=r_k-r$ and using \eqref{eq:dynamics}, we have
\begin{equation}
    \dot x_4=\frac{u_2}{J_k}+f_2(t,\x). \label{eq:x4_dyn}
\end{equation}
where $f_2(t,\x)=-\dot r-\frac{J_{ay}}{J_k}p_kq_a=-\dot r-\frac{J_{ay}}{J_k}(p\cos x_3+q\sin x_3)(-p\sin x_3+q\cos x_3+x_2)$ is obtained by differentiating \eqref{eq:k2a}. Thus the state space model is represented by the dynamics equations 
\begin{equation}
    \dot\x=\begin{bmatrix}
        x_2\\
        \frac{u_1}{J_{ay}}+f_1(t,\x)\\
        x_4\\
        \frac{u_2}{J_k}+f_2(t,\x)
    \end{bmatrix}.\label{eq:ssm}
\end{equation}
In this state space representation of the system, the body fixed angular velocities of the platform, being measured from the rate gyro placed on the platform, are assumed to be known quantities and are thus considered as functions of time. Thus, the state dynamics can be written as a function of time, states and controls.
\begin{remark} \label{rem:noise}
Under the assumptions outlined in \eqref{eq:assum}, significant reductions in inertial disturbances are evident. However, the products of inertia terms may not be entirely eliminated due to potential design errors. One effective way to circumvent this issue is by modifying the functions $f_1$ and $f_2$ to accommodate the inertia disturbance terms that arise due to products of inertia on the RHS of the dynamics equations \eqref{eq:x2_dyn} and \eqref{eq:x4_dyn}. However, this has been avoided to keep the expressions concise as it does not aid in any conceptual development. In the case that the products of inertia terms are relatively insignificant to the moments of inertia, which is often the scenario, these disturbances can be treated as noise, akin to other sources of mechanical noise in the system. In Section \ref{sec:sim}, the results of the proposed control laws are demonstrated in the presence of noise.
\end{remark}

Before looking at the feedback control law to obtain the control objectives of stabilization and tracking, first a transformation of the original system is considered to simplify the dynamics. These transformation is given in the next result where a change of variables is studied that transforms the nonlinear state space model \eqref{eq:ssm} into a linear time-invariant (LTI) system. 
  
\begin{lemma}
    \label{lem:transformation}
    The controls 
\begin{equation}
   \begin{split}
    u_1&=J_{ay}[v_1-f_1(t,\x)] \\
    u_2&=J_k[v_2-f_2(t,\x)] \\   \label{eq:lin_control}
   \end{split}
\end{equation}
are well defined and transform the state space model considered in \eqref{eq:ssm} into a linear dynamical system given by 
\begin{equation}
    \dot\x=A\x+B\v, \label{eq:lin_system}
\end{equation}
where 
\begin{equation}
A=\begin{bmatrix}
    0 & 1 & 0 & 0\\
    0 & 0 & 0 & 0\\
    0 & 0 & 0 & 1\\
    0 & 0 & 0 & 0 
\end{bmatrix},\ B=\begin{bmatrix}
    0 & 0\\
    1 & 0\\
    0 & 0\\
    0 & 1
\end{bmatrix} \text{ and } \v=\begin{bmatrix}
    v_1\\ v_2
\end{bmatrix}. \label{eq:LTI}
\end{equation}
\end{lemma}
\begin{proof}
Consider the controls $\u$ as proposed in \eqref{eq:lin_control}. Using these inputs in the state space model, we have the resulting dynamics given by 
\begin{equation}
\begin{aligned}
    \dot\x&=\begin{bmatrix}
        x_2\\
        \frac{J_{ay}[v_1-f_1(t,\x)]}{J_{a_y}}+f_1(t,x)\\
        x_4\\
        \frac{J_k[v_2-f_2(t,\x)]}{J_k}+f_2(t,x)
    \end{bmatrix},\\
    &\quad=\begin{bmatrix}
        x_2\\
        v_1\\
        x_4\\
        v_2
    \end{bmatrix}.
\end{aligned}
\end{equation}
Thus, the resulting dynamical system can be written as an LTI system with the controls $\v$ and the $A$, $B$ matrices specified by \eqref{eq:LTI}. Similarly, taking the inverse of the transformation considered in \eqref{eq:lin_control} results in 
\begin{equation}
    \begin{aligned}
        v_1&=\frac{u_1}{J_{ay}}+f_1(t,\x),\\
        v_2&=\frac{u_2}{J_k}+f_2(t,\x). \label{eq:inv_trans}
    \end{aligned}   
\end{equation}
Applying these controls to the LTI system \eqref{eq:lin_system}, results in the original state space system \eqref{eq:ssm}.

\end{proof}
\begin{remark}
\label{rem:transformation}
    In Lemma \ref{lem:transformation}, it is shown that the original state space model \eqref{eq:ssm} is equivalent to an LTI system given by \eqref{eq:lin_system}. Thus, any behaviour that can be obtained in the LTI system \eqref{eq:lin_system} can be reproduced in the original model \eqref{eq:ssm} only via a change of variables \eqref{eq:lin_control}. 
\end{remark}

\section{Feedback Control Law}
\label{sec:control_law}
In this section, feedback control laws are proposed to achieve the control objectives of stabilization and tracking. To this end, first consider the definition of two functions.
\begin{definition}
    Consider the functions $g_1,g_2:\mathbb R^+\times\mathcal{X}\rightarrow\mathbb R$ defined by \label{def:g_funct}
    \begin{equation}
        \begin{aligned}
            g_1(t,\x)&=-\dot{p}\sin{x_3}-x_4p\cos{x_3}\\
            &\quad+\dot{q}\cos{x_3}- x_4q\sin{x_3},\text{ and }  \\
            g_2(t,\x)&=(\dot{p}\cos{x_3}-x_4p\sin{x_3}+\dot{q}\sin{x_3}\\
            &+x_4q\cos{x_3})\sin{x_1} + (p\cos{x_3}\\
            &+q\sin{x_3})x_2\cos{x_1}-x_2r\sin{x_1}\\
            &-x_2x_4\sin{x_1}+\dot{r}\cos{x_1}.
        \end{aligned} \label{eq:g_func}
    \end{equation}
    Here onwards the usage of the terms $g_1$ and $g_2$ indicate the functions defined in \eqref{eq:g_func} with an implicit dependence on time and the states. 
\end{definition}
The next result provides a control law that enables the sensor  pitch ($q_a$) and yaw ($r_a$) angular velocities to follow any twice differentiable desired rate trajectories $q_a^d(t)$ and $r_a^d(t)$ respectively.
\begin{theorem}
    Consider the LTI system given by \eqref{eq:lin_system}, \eqref{eq:LTI}. Let $q_a^d(t)$ and $r_a^d(t)$ be once differentiable desired trajectories and let $c_1,\ c_2$ be positive real quantities, then the controls 
    \begin{equation}
        \begin{aligned}
            v_1&=-g_1+\dot q_a^d+c_1(q_a^d-q_a),\\
            v_2&=\frac{-g_2+\dot r_a^d+c_2(r_a^d-r_a)}{\cos x_1},   \label{eq:stab_con}
        \end{aligned}
    \end{equation}
    where $g_1$ and $g_2$ are as introduced in Definition \ref{def:g_funct}, enable the sensor angular velocities $q_a$ and $r_a$ to converge exponentially to the desired rate trajectories $q_a^d$ and $r_a^d$, with decay rates $-c_1$ and $-c_2$ respectively.
    \label{thm:stab}
\end{theorem}

\begin{proof}
    Consider the linear state space model \eqref{eq:lin_system}, \eqref{eq:LTI}, then the  angular velocities $q_a$ and $r_a$ can be written as functions of time and states from the equations \eqref{eq:b2k} and $\eqref{eq:k2a}$ as follows:
    \begin{equation}
        \begin{aligned}
            q_a&=-p\sin x_3+q\cos x_3+x_2,\\
            r_a&=p_k\sin x_1+r_k\cos x_1\\
               &=p\cos x_3\sin x_1+q\sin x_3\sin x_1\\
               &\quad+r\cos x_1+x_4\cos x_1.
        \end{aligned}        
    \end{equation}
    Differentiating $q_a$ we get
    \begin{equation}
   \begin{split}
      \dot{q}_a &=-\dot{p}\sin{x_3}-x_4p\cos{x_3} \\&
                            +\dot{q}\cos{x_3}- x_4q\sin{x_3}+\dot{x_2} \\&
                          = g_1(t,\x) + \dot{x}_2\\
                          &=g_1(t,\x)+v_1,
   \end{split} 
\end{equation}
    while differentiating $r_a$ we get 
    \begin{equation}
    \begin{split}
        \dot{r}_a =& (\dot{p}\cos{x_3}-x_4p\sin{x_3}+
                    \dot{q}\sin{x_3} \\&
                    +x_4q\cos{x_3})\sin{x_1} + (p\cos{x_3} \\&+q\sin{x_3})x_2\cos{x_1}
                    -x_2r\sin{x_1} \\&-x_2x_4\sin{x_1}
                    +\dot{r}\cos{x_1}+\dot{x}_4\cos{x_1} \\
                           &= g_2(t,\x) +\dot{x}_4\cos{x_1} \\
                           &= g_2(t,\x)+v_2\cos x_1.
    \end{split}
\end{equation}
    Plugging in the controls proposed in \eqref{eq:stab_con}, we get the output dynamics as 
    \begin{equation}
        \begin{aligned}
            \dot q_a=\dot q_a^d+c_1(q_a^d-q_a),\\ 
            \dot r_a=\dot r_a^d+c_2(r_a^d-r_a).
        \end{aligned}
    \end{equation}
    Then the errors of the actual angular velocity with the desired trajectories $e_q:=q_a^d-q_a$ and $e_r:=r_a^d-r_a$ follow the following first order dynamics:
    \begin{equation}
        \begin{aligned}
            \dot e_q+c_1e_q=0,\\
            \dot e_r+c_2e_r=0.
        \end{aligned}
    \end{equation}
    By any choice of the control parameters $c_1$ and $c_2$ as positive real values leads to the error dynamics being exponentially stable. Thus, the angular velocities $q_a$ and $r_a$ converge to the desired rate trajectories exponentially fast. 
\end{proof}
\begin{remark}
\label{rem:parameter_choice}
    While the choice of control parameters $c_1$ and $c_2$ may seem arbitrary, it does have a physical interpretation. As shown in Theorem \ref{thm:stab}, these control parameters denote the decay rate of the angular velocity to the desired trajectory. Thus, the choice for these parameters can be made as per the required decay rate considering a trade-off on the demand required on the control input.
\end{remark}
An immediate corollary following Theorem \ref{thm:stab} is that the control law necessary for obtaining LOS stabilization is a special case of the controls given in \eqref{eq:stab_con}. 
\begin{corollary}
    Consider the LTI system given by \eqref{eq:lin_system}, \eqref{eq:LTI}. Let $c_1$ and $c_2$ be positive real quantities, then the controls
    \begin{equation}
        \begin{aligned}
            v_1=-g_1-c_1q_a,\\
            v_2=\frac{-g_2-c_2r_a}{\cos x_1}, \label{eq:stab_cor_con}
        \end{aligned}
    \end{equation}
    where $g_1$ and $g_2$ are as introduced in Definition \ref{def:g_funct}, achieve LOS stabilization by driving the sensor angular velocities $q_a$ and $r_a$ to zero at an exponential rate of $-c_1$ and $-c_2$ respectively. \label{cor:stab}
\end{corollary}
\begin{proof}
    This result is obtained by substituting the desired trajectories $q_a^d(t)$ and $r_a^d(t)$ as identically equal to zero in Theorem \ref{thm:stab}. 
\end{proof}
Following the design of control laws that achieve stabilization of the sensor LOS exponentially fast, the next result proposes a control law that achieves LOS tracking.
\begin{theorem}
    \label{thm:track}
    Consider the LTI system given by \eqref{eq:lin_system}, \eqref{eq:LTI}. Let $\theta_q^d(t)$ and $\theta_r^d(t)$ be twice-differentiable desired trajectory of the elevation ($\theta_q$) and the azimuth ($\theta_r$) angles of the sensor LOS. Let $c_1,\ c_2,\ c_3,\ c_4$ be positive real quantities, then consider the controls 
    \begin{equation}
        \begin{aligned}
            v_1&=-g_1(t,\x)+\ddot\theta_q^d+ c_1(\dot\theta_q^d-\dot\theta_q)+c_2(\theta_q^d-\theta_q), \\
            v_2&=\frac{1}{\cos x_1}\left(-g_2(t,\x)+\ddot\theta_r^d \right.\\
            &\qquad\qquad\left.+c_3(\dot\theta_r^d-\dot\theta_r)+c_4(\theta_r^d-\theta_r)\right), \label{eq:track_con}
        \end{aligned}
    \end{equation}
    where $g_1$ and $g_2$ are as introduced in Definition \ref{def:g_funct}. Then the elevation ($\theta_q$) and azimuth ($\theta_r$) angles of the sensor LOS converge exponentially to the desired trajectories $\theta_q^d(t)$ and $\theta_r^d(t)$ respectively.
\end{theorem}

\begin{proof}
    Consider the linear state space model \eqref{eq:lin_system}, \eqref{eq:LTI}. The double derivatives of the elevation and azimuth angles of the sensor LOS can be written as $\ddot\theta_q=\dot q_a$ and $\ddot\theta_r=\dot r_a$. From the proof of Theorem \ref{thm:stab}, it is clear that $\dot q_a=g_1+v_1$ and $\dot r_a=g_2+v_2\cos x_1$. Now plugging in the controls proposed in \eqref{eq:track_con}, we get the second order dynamics of the outputs as
    \begin{equation}
        \begin{aligned}
            \ddot\theta_q&=\ddot\theta_q^d+ c_1(\dot\theta_q^d-\dot\theta_q)+c_2(\theta_q^d-\theta_q),\text{ and }\\
            \ddot\theta_r&=\ddot\theta_r^d+ c_3(\dot\theta_r^d-\dot\theta_r)+c_4(\theta_r^d-\theta_r).
        \end{aligned}
    \end{equation}
    Denote $e_q:=\theta_q^d-\theta_q$ and $e_r:=\theta_r^d-\theta_r$. Then from the equations above, it is evident that the errors of $\theta_q$ and $\theta_r$ from their respective desired trajectories follow a second order dynamics given by 
    \begin{equation}
        \begin{aligned}
            &\ddot{e}_q+c_1\dot{e}_q+c_2e_q=0 \\
            &\ddot{e}_r+c_3\dot{e}_r+c_4e_r=0
        \end{aligned}
    \end{equation}
    By a suitable choice of the control parameters $c_1,c_2,c_3$ and $c_4$, the decay rate of the errors can be desirably controlled. Thus, this control law drives the sensor LOS towards its desired trajectory and can do so fast enough by making a large enough choice of the control parameters.
\end{proof}
\begin{remark}
\label{rem:stab}
    As mentioned in Remark \ref{rem:transformation}, the behaviour observed in the LTI system \eqref{eq:lin_system} using the controls \eqref{eq:stab_con} or \eqref{eq:track_con}, can be reproduced in the original state space model \eqref{eq:ssm} by using controls $\u$ as given by \eqref{eq:lin_control}. Further, the choice of control parameters can be made by considering a trade-off between the desired decay rate and consequent demands on control input as already discussed in Remark \ref{rem:parameter_choice}.
\end{remark}
\begin{remark}
\label{rem:singularity}
    Note that, the control $v_2$ cannot directly affect $r_a$ or consequently $\theta_r$. This control is achieved through the control of $r_k$ which then affects $r_a$ through a gain of $\cos x_1$. Hence, in both the stabilizing \eqref{eq:stab_con}, \eqref{eq:stab_cor_con} and the tracking \eqref{eq:track_con} controls, there is a $\cos x_1$ term in the denominator for the expression of $v_2$. Thus, as the loop gain ($\cos x_1$) approaches zero, the demand on the control $v_2$ blows up which appears analytically as the cosine term in the denominator.
\end{remark}

\section{Simulations}
\label{sec:sim}
The proposed state space model along with the feedback control laws have been implemented in MATLAB R2023a. The two-axes gimbal system considered is assumed with the following inertia matrices:
\begin{equation} 
         J_A=\begin{bmatrix}
                        0.003 & 0 & 0\\
                        0 & 0.008 & 0\\
                        0 & 0 & 0.003
                \end{bmatrix}
\end{equation}
\begin{equation}
          J_K=\begin{bmatrix}
                    0.003 & 0 & 0\\
                    0& 0.006 & 0\\
                    0 & 0 & 0.0003
                \end{bmatrix}
\end{equation}
Note that these inertia matrices adhere to the assumptions of symmetric design. 
The roll, pitch and yaw angular velocities of the platform are assumed to be sinusoidally varying curves given by $p(t)= 0.1\sin(\frac{\pi}{15}t),\ q(t)=0.1\sin(\frac{\pi}{20}t),\text{ and } r(t)=0.2\sin(\frac{\pi}{15}t) $. This can be seen in Figure \ref{fig:platform_motion}. The cosine term in the denominator of the control $v_2$ has been dealt with a saturation block by restricting its value beyond a threshold around zero.
\begin{figure}[h]
    \centering
    \includegraphics[scale=0.8]{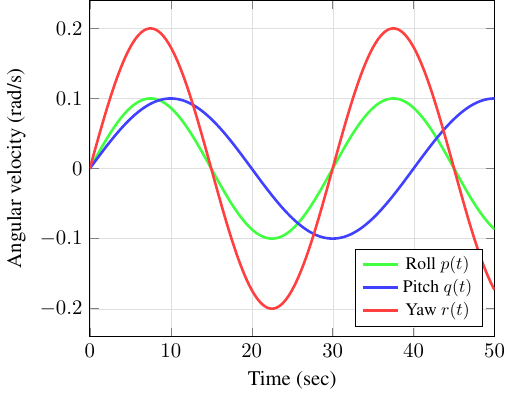}
    \caption{Platform body motion}
    \label{fig:platform_motion}
\end{figure}

\begin{example}
    As already discussed, to achieve stabilization of sensor LOS, the output angular velocities of the pitch channel, $q_a$ and of the yaw channel, $r_a$ are driven to zero. These responses can be seen in Figure \ref{fig:stabilization_plots}. The control parameters used for this case are $c_1=3$ and $c_2=4$. Another case with the presence of noise has been considered. Using the same control parameters however did not yield a stabilizable result. But, by cranking up the control parameters to $c_1=20$ and $c_2=16$, stabilization of sensor LOS can be achieved as shown in Figure \ref{fig:stabilization_plots_noise}.
\end{example} 

\begin{figure}[h]
    \centering
    \begin{subfigure}[h]{0.2\textwidth}
        \includegraphics[scale=0.5]{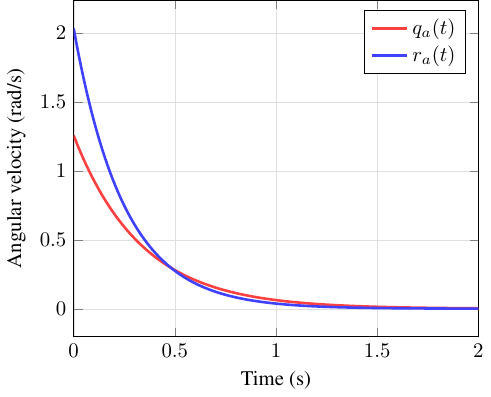}
        \caption{Angular velocities without noise}
        \label{fig:stabilization_plots}
    \end{subfigure}  
    \hspace{20pt}
    \begin{subfigure}[h]{0.2\textwidth}
        \includegraphics[scale=0.5]{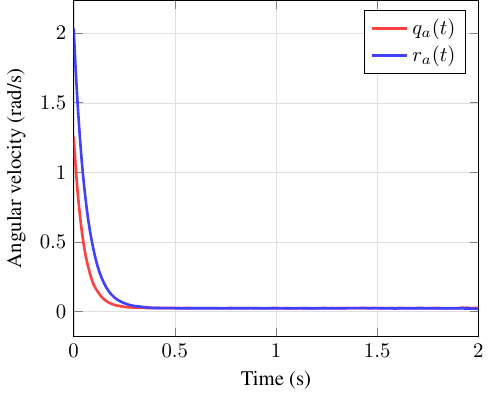}
        \caption{Angular velocities with noise}
        \label{fig:stabilization_plots_noise}
    \end{subfigure}    
    \caption{Angular velocity behaviour during stabilization}
\end{figure}

\begin{example}
    Consider a situation where the desired trajectories for the azimuth and elevation angles are given by step functions. The aim of the azimuth and elevation angles is to reach a step value of $\frac{\pi}{6}$ and $\frac{\pi}{3}$ for a period of $20s$ before dropping back to the initial value. The behaviour of the sensor LOS orientation can be observed in Figures \ref{fig:track_cons_q} and \ref{fig:track_cons_r} in comparison with the desired. A comparison of the proposed controller with a standard PID controller can be made from Figures \ref{fig:track_const_qpid} and \ref{fig:track_const_rpid}. It is evident that while the PID controller manages to smoothly track the step signal, a significant time is required for the LOS orientation to settle at the desired trajectory which can be undesirable in many applications requiring fast and precise tracking. Further, the proposed controller works well with mechanical noise for the both the angles as shown in Figures  \ref{fig:track_cons_qn} and \ref{fig:track_cons_rn}. These responses have been obtained by a use of the control parameters: $c_1=6,\ c_2=8, c_3=9,$ and $c_4=10$.
\end{example}
\begin{figure}[h]
    \centering
    \begin{subfigure}[t]{0.2\textwidth}
       \includegraphics[scale=0.45]{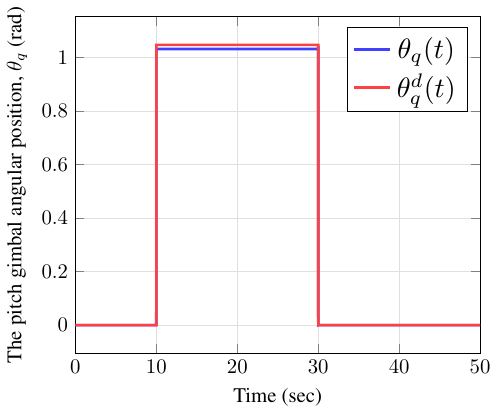}
        \caption{Elevation angle ($\theta_q$) without noise}
        \label{fig:track_cons_q}     
    \end{subfigure}
    \hspace{10 pt}
    \begin{subfigure}[t]{0.2\textwidth}
        \includegraphics[scale=0.45]{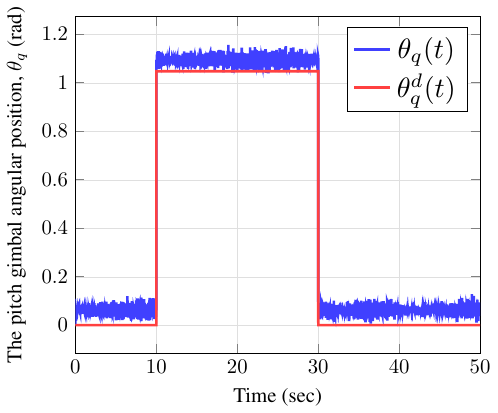}
        \caption{Elevation angle ($\theta_q$) with noise}
        \label{fig:track_cons_qn}
    \end{subfigure}\\
    \begin{subfigure}[t]{0.2\textwidth}
       \includegraphics[scale=0.45]{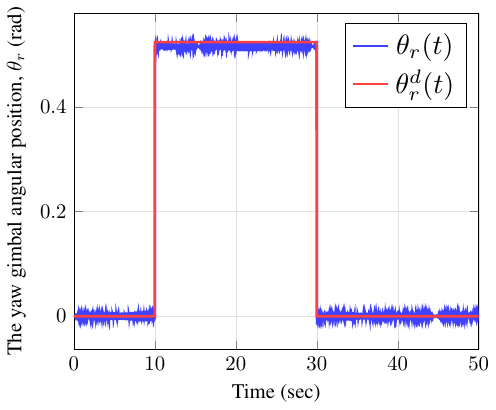}
        \caption{Azimuth angle ($\theta_r$) without noise}
        \label{fig:track_cons_r}     
    \end{subfigure}\hspace{10 pt}
    \begin{subfigure}[t]{0.2\textwidth}
        \includegraphics[scale=0.45]{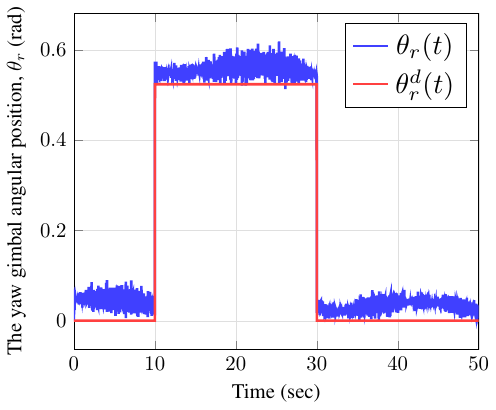}
        \caption{Azimuth angle ($\theta_r$) with noise}
        \label{fig:track_cons_rn}
    \end{subfigure} \\
     \begin{subfigure}[h]{0.2\textwidth}
    	\includegraphics[scale=0.45]{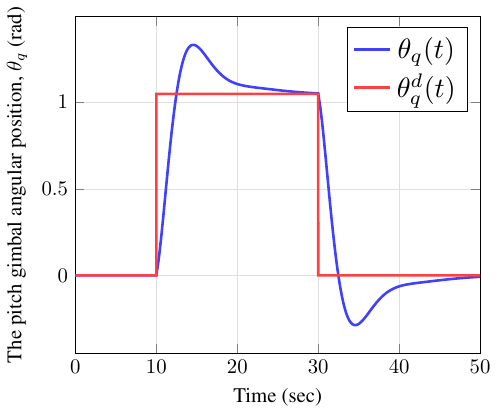}
    	\caption{Elevation angle ($\theta_q$)}
    	\label{fig:track_const_qpid}
    \end{subfigure}  
    \hspace{20pt}
    \begin{subfigure}[h]{0.2\textwidth}
    	\includegraphics[scale=0.45]{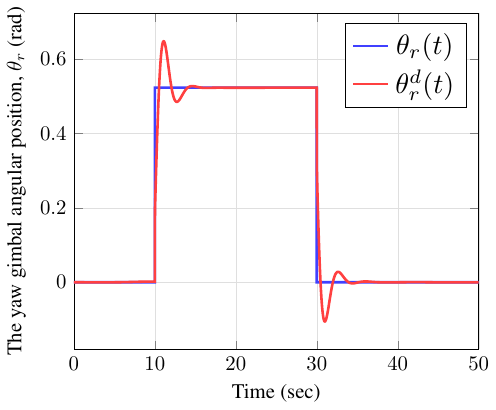}
    	\caption{Azimuth angle ($\theta_r$)}
    	\label{fig:track_const_rpid}
    \end{subfigure} \\
    \caption{Sensor LOS tracking a step signal}
\end{figure}

\begin{example}
    Consider a situation where the desired trajectories for the azimuth and elevation angles are given by sinusoidal functions. This desired trajectory is defined as $\sin(\frac{\pi}{25}t)$. The behaviour of the sensor LOS orientation can be observed in Figures \ref{fig:track_sin_q} and \ref{fig:track_sin_r} in comparison with the desired. A case considering mechanical noise for the both the angles are shown in Figure \ref{fig:track_sin_qn} and \ref{fig:track_sin_rn}. These responses have been obtained by a use of the control parameters: $c_1=8,\ c_2=10, c_3=6,$ and $c_4=8$.
\end{example}

\begin{figure}[h]
    \centering
    \begin{subfigure}[t]{0.2\textwidth}
       \includegraphics[scale=0.45]{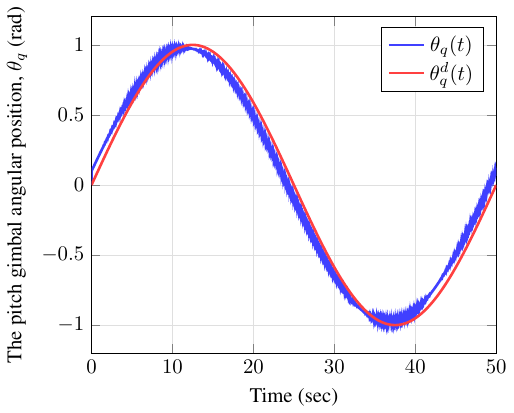}
        \caption{Elevation angle ($\theta_q$) without noise}
        \label{fig:track_sin_q}     
    \end{subfigure}
    \hspace{10 pt}
    \begin{subfigure}[t]{0.2\textwidth}
        \includegraphics[scale=0.45]{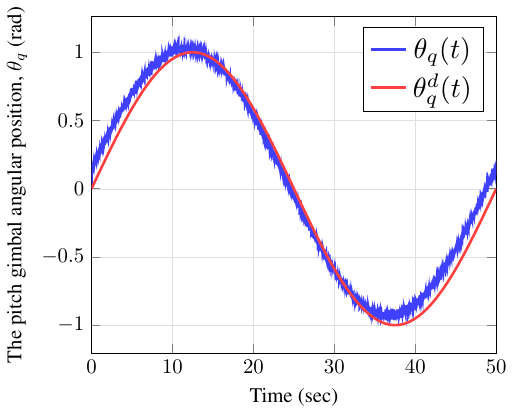}
        \caption{Elevation angle ($\theta_q$) with noise}
        \label{fig:track_sin_qn}
    \end{subfigure}\\
    \begin{subfigure}[t]{0.2\textwidth}
       \includegraphics[scale=0.45]{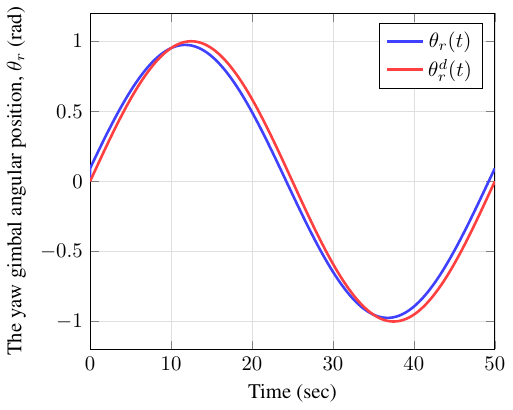}
        \caption{Azimuth angle ($\theta_r$) without noise}
        \label{fig:track_sin_r}     
    \end{subfigure}\hspace{10 pt}
    \begin{subfigure}[t]{0.2\textwidth}
        \includegraphics[scale=0.45]{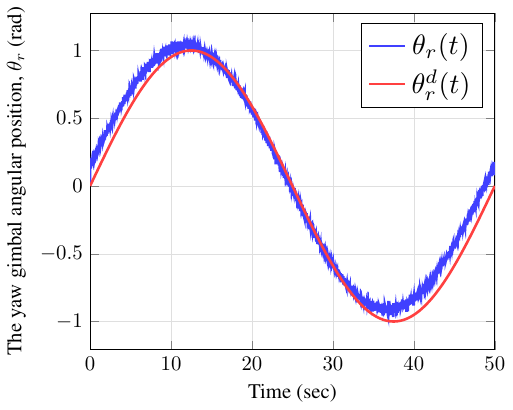}
        \caption{Azimuth angle ($\theta_r$) with noise}
        \label{fig:track_sin_rn}
    \end{subfigure}
    \caption{Sensor LOS tracking a sinusoidal signal}
\end{figure}

\section{Conclusion}
\label{sec:conc}

This paper considers a two-axes gimbal system with symmetry and no mass unbalance. A novel state space model is proposed to capture the system dynamics in a symmetric setting. Using a suitable change of variables, the nonlinear state space model is transformed into a linear time-invariant system. This transformation can also be attained even under asymmetric conditions, and symmetry has been assumed only to keep the expressions concise. Control laws are proposed for the transformed system that enable sensor LOS stabilization and tracking. Further, it has been shown that the control objectives are achieved exponentially fast. The 
 
\bibliographystyle{IEEEtran}
\bibliography{main}

\end{document}